\documentclass[preprint,12pt]{elsarticle}

\usepackage{algorithmic}

\usepackage{amsmath}
\usepackage{amsfonts}
\usepackage{amssymb}
\usepackage{amsthm}
\usepackage{hyperref}
\hypersetup{colorlinks=false}
\newtheorem{theorem}{Theorem}[section]
\newtheorem{lemma}[theorem]{Lemma}

\newtheorem{corollary}[theorem]{Corollary}
\theoremstyle{definition}
\newtheorem{definition}[theorem]{Definition}
\newtheorem{example}[theorem]{Example}
\newtheorem{algorithm}[theorem]{Algorithm}
\theoremstyle{remark}
\newtheorem{remark}[theorem]{Remark}
\numberwithin{theorem}{section} \numberwithin{equation}{section}
\begin{document}
\begin{frontmatter}

\title{A complete algorithm to find exact minimal polynomial by approximations}
\tnotetext[label1]{This research was partially supported by the
National Key Basic Research Project of China (2004CB318003), the
Knowledge Innovation Program of the Chinese Academy of Sciences
(KJCX2-YW-S02), and the National Natural Science Foundation of China
(10771205).}
 \author[label,label1,label2]{Xiaolin Qin}\ead{qinxl@casit.ac.cn}

\author[label]{Yong Feng\corref{cor1}
}\ead{yongfeng@casit.ac.cn}
 \cortext[cor1]{Corresponding author.}
\author[label,label1,label2]{Jingwei Chen}
\author[label,label1]{Jingzhong Zhang}

\address[label]{Lab. of Computer Reasoning and
Trustworthy Comput., University of Electronic Science and Technology
of China, Chengdu 610054, China}

\address[label1]{Lab. for Automated Reasoning and
      Programming, Chengdu Institute of Computer Applications, CAS, Chengdu 610041, China}

\address[label2]{Graduate School of the Chinese Academy of Sciences, Beijing 100049, China}

\begin{abstract}
We present a complete algorithm for finding an exact minimal
polynomial from its approximate value by using an improved
parameterized integer relation construction method. Our result is
superior to the existence of error controlling on obtaining an exact
rational number from its approximation. The algorithm is applicable
for finding exact minimal polynomial of an algebraic number by its
approximate root. This also enables us to provide an efficient
method of converting the rational approximation representation to
the minimal polynomial representation, and devise a simple algorithm
to factor multivariate polynomials with rational coefficients.

Compared with the subsistent methods, our method combines advantage
of high efficiency in numerical computation, and exact, stable
results in symbolic computation. we also discuss some applications
to some transcendental numbers by approximations. Moreover, the
\emph{Digits} of our algorithm is far less than the LLL-lattice
basis reduction technique in theory. In this paper, we completely
implement how to obtain exact results by numerical approximate
computations.
\end{abstract}


\end{frontmatter}

\section{Introduction}
Symbolic computations are principally exact and stable. However,
they have the disadvantage of intermediate expression swell.
Numerical approximate computations can solve large and complex
problems fast, whereas only give approximate results. The growing
demand for speed, accuracy and reliability in mathematical computing
has accelerated the process of blurring the distinction between two
areas of research that were previously quite separate. Therefore,
algorithms that combine ideas from symbolic and numeric computations
have been of increasing interest in the recent two decades. Symbolic
computations are for sake of speed by intermediate use of
floating-point arithmetic. The work reported in
\cite{SSM1991,CG2000,HS2000,CG2001,CAW2002} studied the recovery of
approximate value from numerical intermediate results. A somewhat
related topic is algorithms that obtain the exact factorization of
an exact input polynomial by use of floating point arithmetic in a
practically efficient technique \cite{CG2006}. In the meantime,
symbolic methods are applied in the field of numerical computations
for ill-conditioned problems \cite{C1966,B1971,S1985}. The main goal
of hybrid symbolic-numeric computation is to extend the domain of
efficiently solvable problems. However, there is a gap between
approximate computations and exact results\cite{YZH1992}.

We consider the following question: Suppose we are given an
approximate root of an unknown polynomial with integral coefficients
and a bound on the degree and size of the coefficients of the
polynomial. Is it possible to infer the polynomial and its exact
root? The question was raised by Manuel Blum in Theoretical
Cryptography, and Jingzhong Zhang in Automated Reasoning,
respectively. Kannan \emph{et al} answered the question in
\cite{KLL1988}. However, their technique is based on the
Lenstra-Lenstra-Lovasz(LLL) lattice reduction algorithm, which is
quite unstable in numerical computations. In \cite{Just1989}, Just
\emph{et al} presented an algorithm for finding an integer relation
on $n$ real numbers using the LLL-lattice basis reduction technique,
which needed the high precision. The function
\emph{MinimalPolynomial} in \emph{maple}, which finds minimal
polynomial for an approximate root, was implemented using the same
technique.

In this paper, we present a new algorithm for finding exact minimal
polynomial and reconstructing the exact root by approximate value.
Our algorithm is based on the improved parameterized integer
relation construction algorithm PSLQ($\tau$), whose stability admits
an efficient implementation with lower run times on average than the
existing algorithms, and can be used to prove that relation bounds
obtained from computer runs using it is numerically accurate. The
other function \emph{identify} in \emph{maple} , which finds a
closed form for a decimal approximation of a number, was implemented
using the integer relation construction algorithm. However, the
choice of \emph{Digits} of approximate value is fairly arbitrary
\cite{BHM2003}. In contrast, we fully analyze numerical behavior of
an approximate to exact value and give how many \emph{Digits} of
approximate value, which can be obtained exact results. The work is
regard as a further research in \cite{ZF2007}. We solve the problem,
which can be described as follows:

Given an approximate value $\tilde{\alpha}$ at arbitrary accuracy of
an unknown algebraic number, and we also know the upper bound degree
$n$ of the algebraic number and an upper bound $N$ of
 its height on minimal polynomial in advance.
 The problem will be solved in two steps: First, we discuss how much the error $\varepsilon$ should be,
 so that we can reconstruct the
 algebraic  number $\alpha$ from its approximation
$\tilde{\alpha}$  when it holds that
$|\alpha-\tilde{\alpha}|<\varepsilon$. Of course, $\varepsilon$ is a
function in $n$ and $N$. Second, we give an algorithm to compute the
minimal polynomial of the algebraic number.

Based on our method, we are able to extend our results with the same
techniques to transcendental numbers of the form
sin$^{-1}$($\alpha$), log($\alpha$), etc., where $\alpha$ is
algebraic. we also propose a simple polynomial-time algorithm to
factor multivariate polynomials with rational coefficients, and
provide a natural, efficient technique to the minimal polynomial
representation. The basic idea is from \cite{KLL1988}. However, we
have greatly improved their efficiency.

The rest of this paper is organized as follows. Section 2
illustrates the improved parameterized integer relation construction
algorithm. Section 3 discusses how to reconstruct minimal polynomial
and some applications to some transcendental numbers by
approximations. Section 4 gives some experimental results. The final
section concludes this paper.

This paper is the final journal version of \cite{QFCZ2009}, which
contains essentially the entire contents of this paper.

\section{Preliminaries}
In this section, we first give some notations, and a brief
introduction on integer relation problems. Then an improved
parameterized integer relation construction algorithm is also
reviewed.

\subsection{Notations}
Throughout this paper, $\mathbf{Z}$ denotes the set of the integers,
$\mathbf{Q}$ the set of the rationals, $\mathbf{R}$ the set of the
reals, $\mathbb{O}(\mathbf{R}^{n})$ the corresponding system of
ordinary integers, $U(n-1,\mathbf{R})$ the group of unitary matrices
over $\mathbf{R}$, $GL(n,\mathbb{O}(\mathbf{R}))$ the group of
unimodular matrices with entries in the reals, $col_{i}$B the i-th
column of the matrix B. The ring of polynomials with integral
coefficients will be denoted $\mathbf{Z}[X]$. The $content$ of a
polynomial $p(X)$ in $\mathbf{Z}[X]$ is the greatest common divisor
of its coefficients. A polynomial in $\mathbf{Z}[X]$ is $primitive$
if its content is 1. A polynomial $p(X)$ has degree $d$ if
$p(X)=\sum_{i=0}^{d}p_{i}X^{i}$ with $p_{d}\neq 0$. We write
$deg(p)=d$. The $length$ $|p|$ of $p(X)=\sum_{i=0}^dp_{i}X^{i}$ is
the Euclidean length of the vector $(p_{0},p_{1},\cdots,p_{d})$; the
$height$ $|p|_{\infty}$ of $p(X)$ is the $L_{\infty}$-norm of the
vector$(p_{0},p_{1},\cdots,p_{d})$, so $|p|_{\infty}=\max_{0\leq i
\leq d}|p_{i}|$. An $algebraic$ $number$ is a root of a polynomial
with integral coefficients. The minimal polynomial of an algebraic
number $\alpha$ is the irreducible polynomial in $\mathbf{Z}[X]$
satisfied by $\alpha$. The minimal polynomial is unique up to units
in $\mathbf{Z}$. The $degree$ and $height$ of an algebraic number
are the degree and height of its minimal polynomial, respectively.

\subsection{Integer relation algorithm}
There exists an integer relation amongst the numbers $x_{1}, x_{2},
\cdots, x_{n}$ if there are integers $a_{1}, a_{2}, \cdots, a_{n}$,
not all zero, such that $\sum_{i=1}^n a_{i}x_{i}=0$. For the vector
$\textbf{x}=[x_{1}, x_{2},\cdots, x_{n}]^{T}$, the nonzero vector
$\textbf{a}=[a_{1}, a_{2},\cdots, a_{n}]\in \mathbf{Z}^{n}$ is an
integer relation for $\textbf{x}$ if $\textbf{a}\cdot\textbf{x}=0$.

In order to introduce the integer relation algorithm, we recall some
useful definitions and theorems\cite{FBA1999,BL2000}:
\begin{definition}($M_{\textbf{x}}$)
Assume $\textbf{x}=[x_{1}, x_{2},\cdots, x_{n}]^{T}\in
\mathbf{R}^{n}$ has norm $|\textbf{x}|$=1. Define
$\textbf{x}^{\bot}$ to be the set of all vectors in $\mathbf{R}^{n}$
orthogonal to $\textbf{x}$. Let $\mathbb{O}(\mathbf{R}^{n})\cap
\textbf{x}^{\bot}$be the discrete lattice of integral relations for
$\textbf{x}$. Define $M_{\textbf{x}}>0$ to be the smallest norm of
any relation for $\textbf{x}$ in this lattice.
\end{definition}
\begin{definition}\label{def:Hx}
($H_{\textbf{x}}$) Assume ${\textbf{x}}=[x_{1}, x_{2},\cdots,
x_{n}]^{T}\in \mathbf{R}^{n}$ has norm $|\textbf{x}|$=1.
Furthermore, suppose that no coordinate entry of $\textbf{x}$ is
zero, i.e., $x_{j}\neq 0$ for $1\leq j\leq n$(otherwise $\textbf{x}$
has an immediate and obvious integral relation). For $1\leq j\leq n$
define the partial sums
\begin{eqnarray*}
  s^{2}_{j}=\sum_{j\leq k \leq n}x_{k}^{2}.
\end{eqnarray*}
Given such a unit vector $\textbf{x}$, define the $n\times (n-1)$
lower trapezoidal matrix $H_{\textbf{x}}=(h_{i,j})$ by
$$ h_{i,j}=
\begin{cases}
         0             \ \ \ \  \ \ \ \ \ \ \ \ \ \ \ \ \ \ \ \ \ \ \mbox{if $1\leq i<j\leq n-1,$} \\
         s_{i+1}/s_{i}  \ \ \ \ \ \ \ \ \ \ \ \ \ \ \ \mbox{if $1\leq i=j \leq n-1,$} \\
        -x_{i}x_{j}/(s_{j}s_{j+1}) \ \ \ \ \mbox{if $1 \leq j<i\leq
        n.$} \\
        \end{cases}
$$

Note that $h_{i,j}$ is scale invariant.
\end{definition}
\begin{definition}\label{def:ModifHR}
(Modified Hermite Reduction) Let H be a lower trapezoidal matrix,
with $h_{i,j}=0$ if $j>i$ and $h_{j,j}\neq 0$. Set $D=I_{n}$, define
the matrix $D=(d_{i,j})\in GL(n,\mathbb{O}(\mathbf{R}))$ recursively
as follows: For i from 2 to n, and for j from i-1 to 1(step -1), set
$q=nint(h_{i,j}/h_{j,j})$; then for k from 1 to j replace $h_{i,k}$
by $h_{i,k}-qh_{j,k}$, and for k from 1 to $n$ replace $d_{i,k}$ by
$d_{i,k}-qd_{j,k}$, where the function $nint$ denotes a nearest
integer function, e.g., $nint(t)=\lfloor t+1/2 \rfloor$.
\end{definition}
\begin{theorem}\label{theo:the lower bound}
Let $\textbf{x}\neq 0\in \mathbf{R}^n$. Suppose that for any
relation $\textbf{m}$ of $\textbf{x}$ and for any matrix A $\in
GL(n,\mathbb{O}(\mathbf{R}))$ there exists a unitary matrix Q$\in$
U(n-1) such that $H = AH_{x}Q$ is lower trapezoidal and all of the
diagonal elements of H satisfy $h_{j,j}\neq 0$. Then
 \begin{eqnarray*}
  \frac{1}{\max_{1\leq j \leq n-1}
|h_{j,j}|}=\min_{1\leq j \leq n-1}\frac{1}{|h_{j,j}|}\leq
{|\textbf{m}|}.
\end{eqnarray*}
\end{theorem}
\begin{proof} See Theorem 1 of \cite{FBA1999}.
\end{proof}
\begin{remark}
The inequality of Theorem \ref{theo:the lower bound} offers an
increasing lower bound on the size of any possible relation. Theorem
\ref{theo:the lower bound} can be used with any algorithm that
produces $GL(n,\mathbb{O}(\mathbf{R}))$ matrices. Any
$GL(n,\mathbb{O}(\mathbf{R}))$ matrix $A$ whatsoever can be put into
Theorem \ref{theo:the lower bound}.
\end{remark}
\begin{theorem}\label{theo:iterations}
Assume real numbers, $n\geq 2$, $\tau >1$, $\gamma>\sqrt{4/3}$, and
that $0\neq \textbf{x}\in \mathbf{R}^{n}$ has
$\mathbb{O}(\mathbf{R})$ integer relations. Let $M_{\textbf{x}}$ be
the least norm of relations for $\textbf{x}$. Then PSLQ($\tau$) will
find some integer relation for $\textbf{x}$ in no more than
 \begin{eqnarray*}
 {n\choose 2 }\frac{log(\gamma^{n-1}M_{\textbf{x}})}{log\tau}
 \end{eqnarray*}
 iterations.
\end{theorem}
\begin{proof} See Theorem 2 of \cite{FBA1999}.
\end{proof}
\begin{theorem}\label{theo:the upper bound}
Let $M_{\textbf{x}}$ be the smallest possible norm of any relation
for $\textbf{x}$. Let $\textbf{m}$ be any relation found by
PSLQ($\tau$). For all $\gamma
> \sqrt{4/3}$ for real vectors
 \begin{eqnarray*}
  |\textbf{m}| \leq \gamma^{n-2}M_{\textbf{x}}.
 \end{eqnarray*}
\end{theorem}
\begin{proof} See Theorem 3 of \cite{FBA1999}.
\end{proof}
\begin{remark}
For n=2, Theorem \ref{theo:the upper bound} proves that any relation
$0\neq \textbf{m} \in \mathbb{O}(\mathbf{R}^{2})$ found has norm
$|\textbf{m}|=M_{\textbf{x}}$. In other words, $PSLQ(\tau)$ finds a
shortest relation. For real numbers this corresponds to the case of
the Euclidean algorithm.
\end{remark}
Based on these theorems as above, and if there exists a known error
controlling $\varepsilon$, then an algorithm for obtaining the
integer relation can be designed as follows:
\newcounter{num}
\begin{algorithm}\label{alg:integer relation}
Parameterized
Integer Relation Construction \\
Input: a vector $\textbf{x}$, the upper bound $N$ on the height of
minimal polynomial, and an error $\varepsilon>0$;\\
Output: an integer relation $\textbf{m}$.
\begin{list}{Step \arabic{num}:}{\usecounter{num}\setlength{\rightmargin}{\leftmargin}}
\item Set $i:=1$, $\textbf{m}:=\mathbf{0}$, $\tau>2/\sqrt3$, and unitize the vector $\textbf{x}$ to $\bar{\textbf{x}}$;
\item Set $H_{\bar{\textbf{x}}}$ by definition \ref{def:Hx};
\item Produce matrix $D \in GL(n, \mathbb{O}(\mathbf{R}))$ using modified Hermite
Reduction by definition \ref{def:ModifHR};
\item Set $\bar{\textbf{x}}:=\bar{\textbf{x}}\cdot D^{-1}, H:=D\cdot H, A:=D\cdot A, B:=B\cdot D^{-1}$, \\
case 1: if $\bar{\textbf{x}}_{j}=0$, then $\textbf{m}:=col_{j}B$; \\
case 2: if $h_{i,i}<\varepsilon$, then $\textbf{m}:=col_{n-1}B$;
\item if $0<|\textbf{m}|_{\infty}\leq N$, then goto Step 12;\\
      if $|\textbf{m}|_{\infty}> N$, there is no such an integer
relation, algorithm terminating.
\item $i:=i+1$;
\item Choose an integer $r$, such that $\tau^{r}|h_{r,r}|\geq
\tau^{j}|h_{j,j}|$, for all $1\leq j \leq n-1$;
\item Define $\alpha:=h_{r,r}$, $\beta:= h_{r+1,r}$,
$\lambda:=h_{r+1,r+1}$, $\sigma:=\sqrt{\beta^2+\lambda^2}$;
\item Change $h_{r}$ to $h_{r+1}$, and define the permutation matrix R;
\item Set $\bar{\textbf{x}}:=\bar{\textbf{x}}\cdot R$, $H:= R\cdot H$, $A:=R\cdot A$,
$B:= B\cdot R$, if i=n-1, then goto Step 4;
\item Define $ Q:=(q_{i,j})\in U(n-1,\mathbf{R})$, $H:=H \cdot Q$, goto Step 4;
\item return $\textbf{m}$.
\end{list}
\end{algorithm}

By algorithm \ref{alg:integer relation}, we can find the integer
relation $U(n-1,\mathbf{R})$ of the vector
$\textbf{x}=(1,\tilde{\alpha},\tilde{\alpha}_{2},\cdots,\tilde{\alpha}_{n})$.
So, we get a nonzero polynomial of degree $n$, which denotes $G(x)$
for the rest of this paper, i.e.,
\begin{equation}\label{equ:polynomial}
G(x)=\textbf{m}\cdot(1,x,x^{2},\cdots,x^{n})^{T}.
\end{equation}
Our main task is to show that polynomial (\ref{equ:polynomial}) is
uniquely determined under assumptions, and discuss the controlling
error $\varepsilon$ in algorithm \ref{alg:integer relation} in the
next section.
\section{Reconstructing minimal polynomial from its approximation}
In this section, we will solve such a problem: For a given floating
number $\tilde{\alpha}$, which is an approximation of unknown
algebraic number $\alpha$, how do we obtain the exact value? At
first, we state some lemmas as follows:
\begin{lemma}\label{lem:errorctrol}
Let $f$ be a nonzero polynomial in $\mathbf{Z}[x]$ of degree $n$. If
$\varepsilon=\max_{1\leq i\leq
n}|\alpha^{i}-\tilde{\alpha_{i}}|$\footnote{$\varepsilon$ is defined
by the same way for the rest of this paper.}, where
$\tilde{\alpha}_{i}$ for $1\leq i \leq n$ are the rational
approximations to the powers $\alpha^{i}$ of an algebraic number
$\alpha$, and $\tilde{\alpha}_{0}=1$, then
\begin{equation}\label{equ:errorctrol}
     |f(\alpha)-f(\tilde{\alpha})|\leq \varepsilon \cdot n \cdot
     |f|_{\infty}.
\end{equation}
\end{lemma}
\begin{proof} Clear.
\end{proof}
\begin{lemma}\label{lem:belowbound}
Let $h$ and $g$ be two nonzero polynomials in $\mathbf{Z}[x]$ of
degree $n$ and $m$, respectively, and let $\alpha\in \mathbf{R}$ be
a zero of $h$ with $|\alpha|\leq 1$. If $h$ is irreducible and
$g(\alpha)\neq 0$, then
\begin{equation}\label{equ:belowbound}
     |g(\alpha)|\geq n^{-1}\cdot |h|^{-m}\cdot |g|^{1-n}.
\end{equation}
\end{lemma}
\begin{proof} See Proposition(1.6) of \cite{KLL1988}. If $|\alpha|> 1$, a
simple transform of it does.
\end{proof}
\begin{corollary}\label{cor:extendbelowbound}
Let $h$ and $g$ be two nonzero polynomials in $\mathbf{Z}[x]$ of
degrees $n$ and $m$, respectively, and let $\alpha\in \mathbf{R}$ be
a zero of $h$ with $|\alpha|\leq 1$. If $h$ is irreducible and
$g(\alpha)\neq 0$, then
\begin{equation}\label{equ:belowbound}
     |g(\alpha)|\geq n^{-1}\cdot (n+1)^{-\frac{m}{2}}\cdot (m+1)^{\frac{1-n}{2}}\cdot |h|^{-m}_{\infty}\cdot |g|^{1-n}_{\infty}.
\end{equation}
\end{corollary}
\begin{proof} First notice that $|f|^{2}\leq (n+1)\cdot |f|^{2}_{\infty}$
holds for any polynomial $f$ of degree at most $n>0$, so $|f|\leq
\sqrt{n+1}\cdot |f|_{\infty}$. From Lemma \ref{lem:belowbound}, we
get
\begin{eqnarray*}
     |g(\alpha)|\geq n^{-1}\cdot (n+1)^{-\frac{m}{2}}\cdot (m+1)^{\frac{1-n}{2}}\cdot |h|^{-m}_{\infty}\cdot |g|^{1-n}_{\infty}.
\end{eqnarray*}
This proves Corollary \ref{cor:extendbelowbound}.
\end{proof}
\begin{theorem}\label{theo:the minimal polynomial}
Let $\tilde{\alpha}$ be an approximate value to an unknown algebraic
number $\alpha$ with degree $n>0$, $N$ be the upper bound on the
height of minimal polynomial of $\alpha$. For any $G(x)$ in
$\mathbf{Z}[x]$ with degree n, if
\begin{eqnarray*}
|G(\tilde{\alpha})|<n^{-1}\cdot (n+1)^{-n+\frac{1}{2}}\cdot
|G|^{-n}_{\infty}\cdot N^{1-n}-n\cdot\varepsilon \cdot |G|_{\infty},
\end{eqnarray*}
then
\begin{eqnarray*}
|G(\alpha)|<n^{-1}\cdot (n+1)^{-n+\frac{1}{2}}\cdot
|G|^{-n}_{\infty}\cdot N^{1-n}.
\end{eqnarray*}
\end{theorem}
\begin{proof}
 Let $\alpha\in \mathbf{R}$ be with $|\alpha|\leq 1$. From Lemma
\ref{lem:errorctrol}, we notice that
$|G(\alpha)-G(\tilde{\alpha})|\leq \varepsilon \cdot n \cdot
|G|_{\infty}$, and
$$|G(\alpha)|-|G(\tilde{\alpha})|\leq |G(\alpha)-G(\tilde{\alpha})|,$$ so,
\begin{equation}\label{equ:theorem4}
 |G(\alpha)|\leq|G(\tilde{\alpha})|+n\cdot \varepsilon \cdot
|G|_{\infty}.
\end{equation}
From the assumption of the theorem, since
\begin{equation}\label{equ:max value}
|G(\tilde{\alpha})|<n^{-1}\cdot (n+1)^{-n+\frac{1}{2}}\cdot
|G|^{-n}_{\infty}\cdot N^{1-n}-n\cdot\varepsilon \cdot |G|_{\infty},
\end{equation}
combined with (\ref{equ:theorem4}), we have proved Theorem
\ref{theo:the minimal polynomial}.
\end{proof}
\begin{corollary}\label{cor:minimalpoly}
Let $\tilde{\alpha}$ be an approximate value to an unknown algebraic
number $\alpha$ with degree $n>0$, $N$ be the upper bound on the
height of minimal polynomial of $\alpha$. For any $G(x)$ in
$\mathbf{Z}[x]$ with degree $n$, if $|G(\alpha)|<n^{-1}\cdot
(n+1)^{-n+\frac{1}{2}}\cdot |G|^{-n}_{\infty}\cdot N^{1-n}$, then
\begin{equation}\label{equ:belowbound}
    G(\alpha)=0.
\end{equation}
The primitive part of polynomial $G(x)$ is the minimal polynomial of
algebraic number $\alpha$.
\end{corollary}
\begin{proof} (Proof by Contradiction) Let $\alpha\in \mathbf{R}$ be with
$|\alpha|\leq 1$. According to Lemma \ref{lem:belowbound}, suppose
on the contrary that $G(\alpha)\neq 0$, then
$$|G(\alpha)|\geq n^{-1}\cdot (n+1)^{-n+\frac{1}{2}}\cdot
|G|^{-n}_{\infty}\cdot N^{1-n}.$$ From the assumption of the
corollary, we have
$$|G(\alpha)|<n^{-1}\cdot
(n+1)^{-n+\frac{1}{2}}\cdot |G|^{-n}_{\infty}\cdot N^{1-n}.$$
However, it leads to a contradiction. So, $G(\alpha)=0$.\\
Let $G(x)=\sum_{i=0}^{n}{a_ix^i}$, which is constructed by the
parameterized integer relation construction algorithm from the
vector $\mathrm{x}=(1,{\alpha},{\alpha}^{2},\cdots,{\alpha}^{n})$.
Since algebraic number $\alpha$ with degree $n>0$, according to the
definition of minimal polynomial, then the primitive polynomial of
$G(x)$, denoted by $pp(G(x))$. Hence $pp(G(x))$ is just irreducible
and equal
to $g(x)$. Of course, it is unique.\\
This proves Corollary \ref{cor:minimalpoly}.
\end{proof}

\subsection{Obtaining minimal polynomial by approximation}
If $\alpha$ is a real number, then by definition $\alpha$ is
algebraic exactly, for some $n$, the vector
\begin{equation}\label{equ:the vector}
(1,\alpha,\alpha^{2},\cdots,\alpha^{n})
\end{equation}
has an integer relation. The integral coefficients polynomial of
lowest degree, whose root an algebraic number $\alpha$ is, is
determined uniquely up to a constant multiple; it is called the
$minimal$ $polynomial$ for $\alpha$. Integer relation algorithm can
be employed to search for minimal polynomial in a straightforward
way by simply feeding them the vector (\ref{equ:the vector}) as
their input. Let $\tilde{\alpha}$ be an approximate value belonging
to an unknown algebraic number $\alpha$, considering the vector
$\textbf{v}=(1,\tilde{\alpha},\tilde{\alpha}^{2},\cdots,\tilde{\alpha}^{n})$,
how to obtain the exact minimal polynomial from its approximate
value? We have the same technique answer to the question from the
following theorem.
\begin{theorem}\label{theo:the n-th algebraic number }
Let $\tilde{\alpha}$ be an approximate value belonging to an unknown
algebraic number $\alpha$ of degree $n>0$. If
\begin{equation}\label{equ:n-th maxmin}
\varepsilon=|\alpha-\tilde{\alpha}|<1/(n^{2}(n+1)^{n-\frac{1}{2}}N^{2n}),
\end{equation}
where $N$ is the upper bound on the height of its minimal
polynomial, then $G(\alpha)=0$, and the primitive part of $G(x)$ is
its minimal polynomial.
\end{theorem}
\begin{proof}
 Let $\alpha\in \mathbf{R}$ be with $|\alpha|\leq 1$. From Theorem
\ref{theo:the minimal polynomial} and Corollary
\ref{cor:minimalpoly}, it is obvious that
$$G(\alpha)=0,$$
if and only if
\begin{equation}\label{equ:least bound}
|G(\alpha)|<n^{-1}\cdot (n+1)^{-n+\frac{1}{2}}\cdot
|G|^{-n}_{\infty}\cdot N^{1-n}.
\end{equation}
Under the assumption of the theorem, we get the upper bound of
degree $n$ and an approximate value $\tilde{\alpha}$
belonging to an unknown algebraic number $\alpha$.\\
For substituting the approximate value $\tilde{\alpha}$ in $G(x)$, denoted by $G(\tilde{\alpha })$, there are two cases: \\
Case 1: $G(\tilde{\alpha})\neq0$, $|G(\tilde{\alpha})|>0$. We have
the inequality (\ref{equ:max value})holds, i.e.,
\begin{equation}\label{equ:inequalitysolve}
0<n^{-1}\cdot(n+1)^{-n+\frac{1}{2}}\cdot |G|^{-n}_{\infty}\cdot
N^{1-n}-n\cdot\varepsilon \cdot |G|_{\infty}.
\end{equation}
Clearly, the inequality (\ref{equ:inequalitysolve}) satisfies from the condition (\ref{equ:n-th maxmin}). This proves the Case 1. \\
Case 2: $G(\tilde{\alpha})=0$. From Lemma \ref{lem:errorctrol}, we
have
$|G(\alpha)-G(\tilde{\alpha})|<n\cdot\varepsilon\cdot|G|_{\infty}$,
hence $|G(\alpha)|<n\cdot\varepsilon\cdot|G|_{\infty}$.
 In order to satisfy condition (\ref{equ:least bound}), we only need the following inequality holds,
\begin{equation}\label{equ:epsion2}
 n\cdot\varepsilon \cdot |G|_{\infty}<
n^{-1}\cdot(n+1)^{-n+\frac{1}{2}}\cdot |G|^{-n}_{\infty}\cdot
N^{1-n}.
\end{equation}
From theorem \ref{theo:the upper bound}, and algorithm
\ref{alg:integer relation} in Step 5, $|G|_{\infty}$ is not more
than $N$. Hence we replace $|G|_{\infty}$ by $N$. So the correctness
of the inequality (\ref{equ:epsion2}) follows from (\ref{equ:n-th
maxmin}). This proves Theorem \ref{theo:the n-th algebraic number }.
\end{proof}
It is easiest to appreciate the theorem by seeing how it justifies
the following algorithm for obtaining minimal polynomials from its
approximation:
\begin{algorithm}\label{alg:n-th algebraic}Reconstructing Minimal Polynomial \\
Input: a floating number ($\tilde{\alpha}$, $n$, $N$) belonging to
an unknown  algebraic number $\alpha$,
i.e., satisfying (\ref{equ:n-th maxmin}).\\
Output: $g(x)$, the minimal polynomial of $\alpha$.
\begin{list}{Step \arabic{num}:}{\usecounter{num}\setlength{\rightmargin}{\leftmargin}}
\item Construct the vector $v$;
\item Compute $\varepsilon$ satisfying (\ref{equ:n-th maxmin});
\item Call algorithm \ref{alg:integer relation} to find an integer
relation $w$ for $v$;
\item Obtain $w(x)$ the corresponding polynomial;
\item Evaluate the primitive part of $w(x)$ to $g(x)$;
\item return $g(x)$.
\end{list}
\end{algorithm}
\begin{theorem}
Algorithm \ref{alg:n-th algebraic} works correctly as specified and
uses O($n(logn+logN)$) binary bit operations, where  $n$ and $N$ are
the degree and height of its minimal polynomial, respectively.
\end{theorem}
\begin{proof}
 Correctness follows from Theorem \ref{theo:the n-th
algebraic number }. The cost of the algorithm is O($n(logn+logN)$)
binary bit operations obviously.
\end{proof}

\subsection{Some applications}
In this subsection, we discuss some applications to the
practicalities. The method of obtaining exact minimal polynomial
from an approximate root can be extended to the set of complex
numbers and many applications in computer algebra and science.

This yields a simple factorization algorithm for multivariate
polynomials with rational coefficients: We can reduce a multivariate
polynomial to a bivariate polynomial using the Hilbert
irreducibility theorem, the basic idea was described in
\cite{CAW2002}, and then convert a bivariate polynomial to a
univariate polynomial by substituting a transcendental number in
\cite{MA1985} or an algebraic number of high degree for one variate
in \cite{CFQZ2009}. After this substitution we can get an
approximate root of the univariate polynomial and use our algorithm
to find the irreducible polynomial satisfied by the exact root,
which must then be a factor of the given polynomial. It can find the
bivariate polynomial's factors, from which the factors of the
original multivariate polynomial can be recovered using Hensel
lifting. This is repeated until all the factors are found.

The other yields an efficient method of converting the rational
approximation representation to the minimal polynomial
representation of an algebraic number. The traditional
representation of algebraic numbers is by their minimal polynomials
\cite{BC1990,BCRD1986,EP1997,RL1982}. We now propose an efficient
method to the minimal polynomial representation, which only needs an
approximate value, degree and height of its minimal polynomial,
i.e., an ordered triple $<\tilde{\alpha},n,N>$ instead of an
algebraic number $\alpha$, where $\tilde{\alpha}$ is its approximate
value, and $n$ and $N$ are the degree and height of its minimal
polynomial, respectively, denoted by
$<\alpha>=<\tilde{\alpha},n,N>$. It is not hard to see that the
computations in the representation can be changed to computations in
the other without loss of efficiency, the rational approximation
method is closer to the intuitive notion of computation.

we also discuss some applications to some transcendental numbers by
using an improved parameterized integer relation construction
method. The form of these transcendental numbers is
sin$^{-1}(\alpha)$, cos$^{-1}(\alpha)$, log$(\alpha)$ etc., where
$\alpha$ is an algebraic number. Moreover, a large number of results
were found by using integer relation detection algorithm in the
course of research on multiple sums and quantum field theory in
\cite{DH2000}.

Suppose $\beta$ is the principle value of sin$^{-1}(\alpha)$ for
some unknown $\alpha$, which is, however, known to the algebraic of
degree and height at most $n$ and $N$, respectively. We consider
inferring the minimal polynomial of $\alpha$ from an approximation
$\tilde{\beta}$ to $\beta$ in the deterministic polynomial time. We
show that if $|\beta-\tilde{\beta}|$ is at most
$\varepsilon=1/(n^{2}(n+1)^{n-\frac{1}{2}}N^{2n})$, this can be
done. The specific technique is similar with the method in
\cite{KLL1988}.

Thus, in polynomial time we can compute from $\tilde{\beta}$ an
approximation $\tilde{\alpha}$ to an unknown algebraic number
$\alpha$ such that $|\alpha-\tilde{\alpha}|\leq \varepsilon$, with
$\varepsilon$ as above. Now Theorem (\ref{theo:the n-th algebraic
number }) guarantees that we can find the minimal polynomial of
$\alpha$ in polynomial time.

\section{Experimental results}
 Our algorithms have been implemented in
\emph{Maple}. The following examples run in the platform of Maple 12
under Windows and PIV2.53GB, 512MB of main memory. Table \ref{ta1}
proposes the $Digits$ of approximate values to compare our method
against the LLL-lattice basis reduction algorithm.
\begin{table}[h]
\begin{center}
\begin{tabular}{|c|c|c|c|c|c|}
 \hline $Ex.$ & $n$ & $N$ & $D_{LLL}$ & $D_{PSLQ}$ & $E_{PSLQ}$ \\
\hline 1 & 4 & 13 & 16     &  12     &  8   \\
\hline 2 & 7 & 17 & 36 &  25 &11  \\
\hline 3 & 10 & 15 &57  & 36  & 16 \\
\hline 4 & 15 & 19 &102 & 59 &  25\\
\hline 5 & 23 & 9 & 145 & 77  &  29 \\
\hline 6 & 27 & 19 & 240 &109  & 55\\
\hline 7 & 30 & 15 &277  & 118 & 57 \\
\hline 8 & 34 & 11 & 327 & 126 & 67 \\
\hline 9 & 40 & 15 & 435 & 161   & 82 \\
\hline 10 & 45 & 17 &532  &  189 & 110 \\
\hline 11 & 50 & 13 &620  &200 & 123 \\
\hline 12 & 100 & 13 &2033  &427 & 381 \\
\hline
\end{tabular}
\caption{\label{ta1}Comparison between our algorithm and LLL-lattice
basis reduction technique}\vskip3mm
\end{center}
\vspace{-0.15in}
\end{table}

In Table \ref{ta1}, we present many examples to compare our new
method against the LLL-lattice basis reduction algorithm. For each
example, we construct the irreducible polynomial with random
integral coefficients in the range $-20 \leq c \leq 20$. Here $n$
and $N$ denote the degree of algebraic number and the height of its
corresponding minimal polynomial respectively; whereas $D_{LLL}$ and
$D_{PSLQ}$ are relative $Digits$ to obtain the exact minimal
polynomial in theory, $E_{PSLQ}$ is in our optimal experimental
results respectively.

From Table \ref{ta1}, we observe that the \emph{Digits} of our
algorithm is far less than the LLL-lattice basis reduction technique
in theory. However, the \emph{Digits} of our algorithm is more than
that of the optimal experimental results. So, in the further work we
would like to consider improving the error controlling.

The following first two examples illuminate how to obtain an exact
quadratic algebraic number and minimal polynomial. Example 3 uses a
simple example to test our algorithm for factoring primitive
polynomials with integral coefficients.
\begin{example}
Let $\alpha$ be an unknown quadratic algebraic number. We only know
an upper bound of height on its minimal polynomial $N=47$. According
to theorem \ref{theo:the n-th algebraic number }, compute quadratic
algebraic number $\alpha$ as follows: First obtain control error
$\varepsilon=1/(12*\sqrt{3}*N^{4})=1/(1807729447692*\sqrt{3})\approx
1.0\times10^{-8}$. And then assume that we use some numerical method
to get an approximation $\tilde{\alpha}=11.937253933$, such that
$|\alpha-\tilde{\alpha}|<\varepsilon$. Calling algorithm
\ref{alg:n-th algebraic} yields as follows: Its minimal polynomial
is $g(x)=x^{2}-8*x-47$. So, we can obtain the corresponding
quadratic algebraic number $\alpha=4+3\sqrt{7}$.
\end{example}
\begin{remark}
  The function \emph{identify} in maple 12 needs \emph{Digits}:=13,
  whereas our algorithm only needs 9 digits.
\end{remark}
\begin{example}\label{exam:1}
Let a known floating number $\tilde{\alpha}$ belonging to some
algebraic number $\alpha$ of degree $n=4$, where
$\tilde{\alpha}=3.14626436994198$, we also know an upper bound of
height on its minimal polynomial $N=10$. According to theorem
\ref{theo:the n-th algebraic number }, we can get the error
$\varepsilon=1/(n^{2}(n+1)^{n-\frac{1}{2}}N^{2n})=1/(4^{2}*5^{\frac{7}{2}}*10^{8})\approx2.2\times10^{-12}$.
Calling algorithm \ref{alg:n-th algebraic}, if only the floating
number $\tilde{\alpha}$, such that
$|\alpha-\tilde{\alpha}|<\varepsilon$, then we can get its minimal
polynomial $g(x)=x^{4}-10*x^{2}+1$. So, the exact algebraic number
$\alpha$ is able to be denoted by
$<\alpha>=<3.14626436994198,4,10>$, i.e.,
$<\sqrt{2}+\sqrt{3}>=<3.14626436994198,4,10>$.
\end{example}
\begin{remark}
In the example $\ref{exam:1}$, we only propose a simple example to
represent the exact algebraic number by approximate method. In the
further work, we would like to discuss the efficient arithmetic
algorithms for summation, multiplication and inverse of the rational
approximation representation.
\end{remark}

\begin{example} This example is an application in factoring
primitive polynomials over integral coefficients. For convenience
and space-saving purposes, we choose a very simple polynomial as
follows:
$$p=3x^{9}-9x^{8}+3x^{7}+6x^{5}-27x^{4}+21x^{3}+30x^{2}-21x+3$$
We want to factor the polynomial $p$ via reconstruction of minimal
polynomials over the integers. First, we transform $p$ to a
primitive polynomial as follows:
$$p=x^{9}-3x^{8}+x^{7}+2x^{5}-9x^{4}+7x^{3}+10x^{2}-7x+1,$$
We see the upper bound of coefficients on polynomial $p$ is $10$,
which has relation with an upper bound of coefficients of the
factors on the primitive polynomial $p$ by Landau-Mignotte bound
\cite{M1974}. Taking $N=5$, $n=2$ yields
$\varepsilon=1/(2^{2}*(2+1)^{2-\frac{1}{2}}*5^{4})=1/(7500*\sqrt{3})\approx8.0\times10^{-5}$.
Then we compute the approximate root on $x$. With Maple we get via
[fsolve($p=0,x$)]:\\ $S=[2.618033989, 1.250523220, -.9223475138,
.3819660113, .2192284350].$

According to theorem \ref{theo:the n-th algebraic number }, let
$\tilde{\alpha}=2.618033989$ be an approximate value belonging to
some quadratic algebraic number $\alpha$, calling algorithm
\ref{alg:n-th algebraic} yields as follows: $$p_{1}=x^{2}-3*x+1.$$
And then we use the polynomial division to get
$$p_{2}=x^{7}+2*x^{3}-3*x^{2}-4*x+1.$$
Based on the Eisenstein's Criterion \cite{SL2002}, the $p_{2}$ is
irreducible in $\mathbf{Z}[X]$. So, the $p_{1}$ and $p_{2}$ are the
factors of primitive polynomial $p$.
\end{example}

\section{Conclusions}
In this paper, we propose a new method for obtaining exact results
by numerical approximate computations. The key technique of our
method is based on an improved parameterized integer relation
construction algorithm, which is able to find an exact relation by
the accuracy control. The error $\varepsilon$ in formula
(\ref{equ:n-th maxmin}) is an exponential function in degree and
height of its minimal polynomial. The result is superior to the
existence of error controlling on obtaining an exact rational number
from its approximation in \cite{ZF2007}. Using our algorithm, we
have succeed in factoring multivariate polynomials with rational
coefficients and providing an efficient method of converting the
rational approximation representation to the minimal polynomial
representation. Our method can be applied in many aspects, such as
proving inequality statements and equality statements, and computing
resultants, etc.. Thus we can take fully advantage of approximate
methods to solve larger scale symbolic computation problems.
Furthermore, our basic idea can be generalized easily to complex
algebraic numbers.


\begin{thebibliography}{30}
\providecommand{\natexlab}[1]{#1}
\providecommand{\url}[1]{\texttt{#1}}
\providecommand{\urlprefix}{URL } \expandafter\ifx\csname
urlstyle\endcsname\relax
  \providecommand{\doi}[1]{doi:\discretionary{}{}{}#1}\else
  \providecommand{\doi}[1]{doi:\discretionary{}{}{}\begingroup
  \urlstyle{rm}\url{#1}\endgroup}\fi
\providecommand{\bibinfo}[2]{#2}

\bibitem{BC1990}
Boehm, H.J. and Cartwright, R. Exact real arithmetic: Formulating
real numbers as functions. In Turner. D., editor, \emph{Research
Topics in Functional Programming,} Addison-Wesley, (1990), pp.
43-64.
\bibitem{BCRD1986}
Boehm, H.J., Cartwright, R., Riggle, M., et al. Exact real
arithmetic: A case study in higher order programming. \emph{In ACM
Symposium on Lisp and Functional Programming}, 1986.
\bibitem{BL2000}
Borwein, J. M., and Lisonek, P. Applications of Integer Relation
Algorithms. \emph{Disc. Math.} 217(2000): 65-82.
\bibitem{BHM2003}
Borwein, P., Hare, K. G. and Meichsener, A. Reverse Symbolic
Computations, The \emph{identity} Function. in: Proceedings from the
Maple Summer Workshop, Maple Software, Waterloo, 2002.
\bibitem{B1971}
Brown, W. S. On Euclid's algorithm and the computation of polynomial
greatest common divisors. \emph{J. ACM.} 18, 4(1971)): 478-504.
\bibitem{CG2006}
Ch$\grave{e}$ze, G., Galligo, A. From an approximate to an exact
absolute polynomial factorization. \emph{Journal of Symbolic
Computation,} 41: 682-696, 2006.
\bibitem{CFQZ2009}
Chen, J. W., Feng Y., Qin X. L., et al. Exact Polynomial
Factorization by Approximate High Degree Algebraic Numbers. In Proc.
2009 Internat. Workshop on Symbolic-Numeric Comput. ACM press,
21-28,2009.
\bibitem{C1966}
Collins, G. E. Polynomial remainder sequences and determinants.
\emph{Amer. Math. Monthly, }73(1966): 708-712.
\bibitem{CAW2002}
Corless, R. M., Galligo, A., Kotsireas, I. S., et al. A
geometric-numeric algorithm for absolute factorization of
multivariate polynomials. \emph{In Proc. 2002 Internat. Symp.
Comput. ISSAC'02}, ACM press, pp. 37-45.
\bibitem{CG2000}
Corless, R.M., Giesbrecht, M. W., et al. Numerical implicitization
of parametric hypersurfaces with linear algebra. \emph{In Proc.
AISC2000, LNAI 1930}, pp. 174-183.
\bibitem{CG2001}
Corless, R. M., Giesbrecht, M. W., et al. Towards factoring
bivariate approximate polynomials. \emph{In Proc. 2001 Internat.
Symp. Comput. ISSAC'01}, ACM press, pp. 85-92.
\bibitem{DH2000}
David H. B.. Integer Relation Detection. Computing in Science and
Engineering, 2, 1(2000), pp. 24-28.
\bibitem{EP1997}
Edalat, A. and Potts, P. J. A new representation for exact real
numbers. \emph{In Mathematical foundations of programming semantics
(Pittsburgh, PA, 1997)}, page 14 pp. (electronic). Elsevier,
Amsterdam, 1997.
\bibitem{FBA1999}
Ferguson, H. R. P., Bailey, D. H., and Arno, S. Analysis of PSLQ, An
Integer Relation Finding Algorithm.\emph{ Math. Comput.}, 68,
225(1999): 351-369.
\bibitem{HS2000}
Huang, Y., Wu, W., Stetter, H., et al. Pseudofactors of multivariate
polynomials. \emph{In Proc. 2000 Internat. Symp. Comput. ISSAC'00},
ACM press, pp. 161-168.
\bibitem{Just1989}
Just, B. Integer relations among algebraic numbers. Mathematical
Foundations of Computer Science, (1989), pp. 314-320.
\bibitem{KLL1988} Kannan, R.,
Lenstra, A.K., and Lov$\acute{a}$sz, L. Polynomial Factorization and
Nonrandomness of Bits of Algebraic and Some Transcendental Numbers.
\emph{Math.Comput.} 50, 182(1988): 235-250.
\bibitem{SL2002}
Lang, S. \emph{Algebra}, 3rd edition, Springer-Verlag, New York,
2002.
\bibitem{RL1982}
Loos, B. Computing in Algebraic Extensions, \emph{Computer Algebra,}
(Ed. by B. Buchberger, et al), Springer-Verlag, (1982), pp. 173-187.
\bibitem{M1974}
Mignotte, M., An inequality about factors of polynomials,
\emph{Math. Comp.}, 28, 128(1974): 1153-1157.
\bibitem{QFCZ2009}
 Qin, X. L., Feng, Y., Chen J. W., et al. Finding Exact Minimal Polynomial
  by Approximations. In Proc. 2009 Internat. Workshop on Symbolic-Numeric Comput. ACM press,125-131,2009.
\bibitem{SSM1991}
Sasaki, T., Suzuki, M., et al. Approximate factorization of
multivariate polynomials and absolute irreducibility testing.
\emph{Japan J. Indust. Appl. Math.} 8 (1991), 357-375.
\bibitem{S1985}
Sch$\ddot{o}$nhage, A. Quasi-gcd computations. \emph{J.Complexity}.
1(1985): 118-137.
\bibitem{MA1985}
Van Der Hulst, M.-P. and Lenstra, A. K. Factorization of polynomials
by transcendental evaluation. \emph{EUROCAL'85}, (1985), pp.
138-145.
\bibitem{YZH1992}
Yang, L., Zhang, J. Z., and Hou, X. R.  A criterion of dependency
between algebraic equation and its application. Proc.
\emph{IWMN'92,} Internat. Academic Publishers,  pp. 110-134, 1992.
\bibitem{ZF2007}
Zhang, J. Z., Feng, Y. Obtaining Exact Value by Approximate
Computations.\emph{ Science in China Series A: Mathematics}. 50,
9(2007): 1361-1368.

\end{thebibliography}
   \end{document}